\documentclass[12pt]{iopart}
\usepackage{amsfonts}
\usepackage{graphicx}%Вставка картинок правильная
\usepackage{float}%"Плавающие" картинки
\usepackage{wrapfig}%Обтекание фигур (таблиц, картинок и прочего)
\usepackage[normalem]{ulem}
\usepackage{hyperref}

\newcommand{\ket}[1]{{\left\vert{#1}\right\rangle}}
\newcommand{\braket}[2]{{\langle {#1}\!\mid\!{#2} \rangle}}

\newtheorem{theorem}{Theorem}
\newtheorem{definition}{Definition}

\newtheorem{remark}{Remark}
\newtheorem{property}{Property}

\newenvironment{proof}%
{\medskip
\noindent {\it Proof.}
}{\Endproof}
\newcommand{\Endproof}{\hfill$\Box$\\}

\usepackage{xcolor}

\eqnobysec%equation numbering by section, e.g. (2.1), (2.2), etc.

\begin{document}

\title{Multiqudit quantum hashing and its implementation based on orbital angular momentum encoding}

\author{D O Akat’ev$^{1}$, A V Vasiliev$^{1,2}$, N M Shafeev$^{2}$, F M Ablayev$^{1,2}$, and A A Kalachev$^{1,2}$}

\address{$^1$ Zavoisky Physical-Technical Institute, FRC Kazan Scientific Center of RAS, Kazan, Russian Federation}
\address{$^2$ Kazan Federal University, Kazan, Russian Federation}
\ead{akatevdmitrijj@gmail.com, vav.kpfu@gmail.com}
\vspace{10pt}
\begin{indented}
\item[]October 2022
\end{indented}

%\date{}

%\institute{Kazan Federal University,
%Kazan, Russian Federation \and Institute for Informatics% of Tatarstan Academy of Sciences
%, Kazan, Russian Federation}

%\maketitle

\begin{abstract}%should not normally exceed 200 words
A new version of quantum hashing technique is developed wherein a quantum hash is constructed as
a sequence of single-photon high-dimensional states (qudits). A proof-of-principle implementation of the high-dimensional quantum hashing protocol using orbital-angular momentum
encoding of single photons is implemented. It is shown that the number of qudits decreases with increase of their dimension for an optimal ratio between collision probability and
decoding probability of the hash. Thus, increasing dimension of information carriers makes quantum hashing with single photons more efficient.
\end{abstract}

\section{Introduction}

Hashing algorithms today have become essential in cybersecurity, cryptography, data-intensive research, etc., as they can reliably inform us whether two files are identical without opening and comparing them. As an important part of cryptography, a hash function compresses a message of any length into a digest of fixed length and it is the key technology of verification of message integrity, digital signatures, fingerprinting and other cryptographic applications \cite{Paar-Pelzl:2009:Understanding-Cryptography,Katz-Lindell:2014:Modern-Cryptography}. Moreover, a universal hash function is an important part of the privacy amplification process of the quantum key distribution \cite{476316}. For such applications, a good hashing algorithm should satisfy two main properties: one-way property and collision resistance. The first means that restoring an input from its hash should be a computationally hard problem. The second property means that the situation when two different inputs have the same hash (such a situation is called a collision) is hard to find. 

Recently, a promising generalization of the cryptographic hashing concept on the quantum domain, which is called quantum hashing, has been suggested and developed \cite{Ablayev-Vasiliev:2014:Crypto-Q-Hashing,Ablayev-et-al:2016:Quantum-Fingerprinting-Hashing,Ablayev-Ablayev-Vasiliev:2016:Balanced-Quantum-Hashing,Our_qubits_hash}. In this case, the hash function encodes a classical input state into a quantum state so that to optimize the trade-off between one-way property and collision resistance. In particular, in \cite{Our_qubits_hash}, it was suggested to construct a quantum hash via a sequence of single-photon qubits, and a proof-of-principle experiment using single photons with orbital angular momentum (OAM) encoding was implemented. In the present paper, we further develop this approach both theoretically and experimentally and construct a quantum hash as a sequence of single-photon high-dimensional states (qudits). We show that the use of high-dimensional states increases the collision resistance of hashing protocols and enhances the resistance against extraction of information about the classical input.

\section{Theory}

\subsection{Preliminaries}

%\paragraph{Quantum hashing via small-biased sets.}

In \cite{Ablayev-Vasiliev:2014:Crypto-Q-Hashing} we have proposed a
cryptographic quantum hash function and later in
\cite{Vasiliev:2016:Quantum-Hashing-for-Groups} provided its
generalized version for arbitrary finite abelian groups based on the
notion of $\varepsilon$-biased sets. 
%Note that $\varepsilon$-biased
%sets are generally defined for arbitrary finite groups
%\cite{Chen-Moore-Russell:2013:Small-Bias-Sets-for-Nonabelian-Groups},
%but here we restrict ourselves to $\mathbb Z_q$. 
Here we consider its version for a cyclic group $\mathbb Z_q$.
In this case, for a set  
$S\subseteq \mathbb Z_q$ we can define its \emph{bias} with respect to $x\in \mathbb Z_q$ as following:
\begin{equation}
bias(S,x)=\frac{1}{|S|}\left|\sum_{s\in S}e^{2\pi sx/q}\right|,
\end{equation}
and the set $S$ is called $\varepsilon$-biased if for any $x\neq0$ $bias(S,x)\le\varepsilon$.
% \begin{equation}
% \frac{1}{|S|}\left|\sum_{s\in S}e^{2\pi sx/q}\right|\le\varepsilon.
% \end{equation}

These sets are especially interesting when $|S|\ll
|\mathbb Z_q|$ (as $S=\mathbb Z_q$ is obviously 0-biased). In their
seminal paper \cite{Naor-Naor:1990:Small-Bias-Spaces} Naor and Naor
defined these small-biased sets, gave the first explicit
constructions of such sets, and demonstrated the power of
small-biased sets for several applications. 
Note that $\varepsilon$-biased sets of size $O(\log q/\varepsilon^2)$ exist as proved in
\cite{Alon-Roichman:1994:Random-Cayley-graphs}.

%\paragraph{Quantum Hashing.}

In \cite{Ablayev-Vasiliev:2014:Crypto-Q-Hashing} we have introduced the notion of quantum hashing and its main properties. Later in \cite{Ablayev-Ablayev-Vasiliev:2016:Balanced-Quantum-Hashing} we have considered the trade-off and balancing between two main properties of quantum hashing, and proposed a more general definition of the quantum $(\delta,\varepsilon)$-resistant hash function. Here we recall it in the concise manner and refer for details to \cite{Ablayev-Ablayev-Vasiliev:2016:Balanced-Quantum-Hashing}.

\begin{definition}\label{QHashing_definition}
Let $\delta\in(0,1]$ and $\varepsilon\in[0,1)$.  We call a function
 $ \psi : {\mathbb X} \to {\mathcal H}^K
 $
a quantum $(\delta,\varepsilon)$-resistant hash function if it has two main properties:
\begin{enumerate}
    \item $\delta$-one-wayness, i.e.
    $$\frac{K}{\left|\mathbb X\right|}\leq \delta,$$
    \item $\varepsilon$-collision-resistance, i.e. for any pair $x_1,x_2$ of different inputs
$$
\left|\braket{\psi(x_1)}{\psi(x_2)}\right| \le \varepsilon.
 $$   
\end{enumerate}

\end{definition}

In other words a quantum function $\psi$ encodes an input $x\in \mathbb X$ into the quantum state $\ket{\psi(x)}$ of dimension $K$. The properties of such a function include resistance to inversion (known as ``one-way property'' or ``preimage resistance''), which makes it unlikely to ``extract'' encoded information out of the quantum state, and resistance to quantum collisions, which means that quantum images for different inputs can be distinguished with high probability.

Note that the measure of collision resistance (denoted above by $\varepsilon$) is not the probability of quantum collisions. The probability of collisions follows from the particular comparison procedure that we use. It can be the well-known SWAP-test \cite{Buhrman:2001:Fingerprinting}, REVERSE-test \cite{Ablayev-Ablayev:2015:Quantum-Hashing-LPL} or simply a result of projection of $\ket{\psi(x_1)}$ onto $\ket{\psi(x_2)}$. In the latter case the probability of collisions would be described by the fidelity between $\ket{\psi(x_1)}$ and $\ket{\psi(x_2)}$, and thus bounded by $\varepsilon^2$.

% Note that the collision resistance means near-orthogonality %($\varepsilon$-orthogonality) !!!!!!
% of quantum
% states $\ket{\psi(x_1)}$ and $\ket{\psi(x_2)}$. It is well-known that orthogonality of quantum states
% provides their distinguishability. 
% In the context of quantum functions near-orthogonality
% means high collision resistance. 
% There is a known lower bound by Buhrman et al. \cite{Buhrman:2001:Fingerprinting} for the size of
% the space to accommodate the set of pairwise-distinguishable states: to construct a set of $\left|\mathbb X\right|$ quantum states with
% pairwise inner products below $\varepsilon$ the dimension of the state space need to be at least $\Omega(\log{\left|\mathbb X\right|}/\varepsilon)$. 
% %This implies that an  $\varepsilon$-collision-resistant quantum function requires at least $s=\Omega(\log\log|{\mathbb X}|-\log{\varepsilon}))$ qubits. The similar lower
% %bound of $\log\log{K} - c(\varepsilon)$ was proved by a different method in
% %\cite{Ablayev-Ablayev:2015:e-universal-Quantum-Hashing}.
% Thus
% %, from the Property \ref{preimage-bound} 
% it is follows that an $\varepsilon$-collision-resistant quantum function can be $\delta$-one-way for $\delta=\Omega\left(\log|{\mathbb X}|/
% (\varepsilon|{\mathbb X}|\right)$.

\subsection{Multiqudit Quantum Hashing}\label{quantum-hash-functions}

Here we define a new version of the quantum hashing technique for a cyclic group, i.e. we consider $\mathbb X=\mathbb Z_q$ and $|\mathbb X|=q$. It is based on small-biased sets and high-dimensional states (qudits).
But first we note the following equivalence between $\varepsilon$-biased sets.
\begin{property}\label{E-Biased-Sets-equivalence}
Let $S=\{s_1, \ldots, s_{d}\}$ and $S'=\{0, (s_{2}-s_{1}), \ldots, (s_{d}-s_{1})\}$. Then for any $x\in \mathbb Z_q$ $bias(S,x)=bias(S',x)$, i.e. the set $S$ is equivalent (in terms of its bias) to $S'$.
\end{property}
\begin{proof}
The proof of this statement is based on the following considerations:
\[\fl \frac{1}{d}\left|\sum_{k=1}^{d}e^{2\pi s_{k}x/q}\right|= \frac{1}{d}\left|e^{2\pi s_{1}x/q}\right|\left|\sum_{k=1}^{d}e^{2\pi (s_{k}-s_{1})x/q}\right|=\\
\frac{1}{d}\left|\sum_{k=1}^{d}e^{2\pi (s_{k}-s_{1})x/q}\right|,
\]
therefore
\begin{eqnarray*}
\fl bias(S,x)=\frac{1}{d}\left|\sum_{k=1}^{d}e^{2\pi s_{k}x/q}\right|
=\frac{1}{d}\left|\sum_{k=1}^{d}e^{2\pi(s_{k}-s_{1})x/q}\right|=
bias(S',x).
\end{eqnarray*}
\end{proof}
Now let $S_1, S_2, \ldots, S_m \subset\mathbb Z_q$ be the
$\varepsilon$-biased subsets of $\mathbb Z_q$, and we denote
$S_j=\{s_{j,1}, \ldots, s_{j,d}\}$ for $j=1, \ldots, m$.
By the Property \ref{E-Biased-Sets-equivalence} without loss of generality we may consider all $s_{j,1}$ to be equal 0. In other words for all $j=1, \ldots, m$ it holds that
$$
\max\limits_{x\neq 0} \frac{1}{d}\Big|1+e^{i\frac{2\pi s_{j,2}x}{q}}+\ldots+ e^{i\frac{2\pi s_{j,d}x}{q}}\Big|\leq \varepsilon.
$$
Then for $x\in\mathbb Z_q$ we define a multiqudit quantum hash
function in the following way:
% The complexity of computing quantum hash function $\psi_{H_S}$ in the QOBDD model is given by the following Theorem.
% \begin{theorem}\label{homom}\label{qbp_h} Quantum $(\delta,\varepsilon)$-hash function  (\ref{qhf-bs})
% \[ \psi_{H_S} : {\mathbb F}_q \to ({\cal H}^2)^{\otimes s}  \]
% %that satisfies Property \ref{v2016pract}
% can be computed by   quantum  OBDD  $Q$ composed from $s=O(\log{\log q})$ qubits.% in $\log{q}$ steps.
% \end{theorem}
% {\em Proof.} The quantum function $\psi_{H_S}$ (\ref{qhf-bs}) maps an input $x\in \mathbb F_q$ to a  quantum state (\ref{qs-bs})
\begin{eqnarray}
\label{single-qudit-quantum-hashing}
  \ket{\psi_j(x)} = \frac{1}{\sqrt{d}}(\ket{\ell_1} + e^{i 2\pi s_{j,2}x/q}\ket{\ell_2}+\ldots+e^{i2\pi s_{j,d}x/q}\ket{\ell_d}),
  %\frac{1}{\sqrt{d}}\sum_{k=0}^{d-1}e^{2\pi s_{j,k}x/q}\ket{k},
\end{eqnarray}
\begin{equation}\label{multiqudit-quantum-hashing}
  \ket{\psi(x)} = \ket{\psi_1(x)} \otimes \cdots \otimes
  \ket{\psi_m(x)},
\end{equation}
%Note that $\ket{\psi_j(x)}$ exactly corresponds to the single-photon multidimensional states with an orbital angular momentum encoding mentioned in the previous section.
where $\ket{\ell_k}$ are the basis states ($k=1\ldots d$, $d$ is the dimension of the qudit state space), $q$ is the size of the input state space, $x\in \{0, 1, \ldots, q-1\}$ is a classical input that is encoded by the relative phase of $m$ qudit states, $s_{i,k}$ are numeric parameters (elements of the $\varepsilon$-biased sets) of the quantum hash function that provide its collision resistance. The main idea of the collision resistance property is to provide minimum fidelity between different quantum hashes (quantum hash function images) with the minimal possible number of quantum information carriers. Furthermore, reaching reasonable balance between collision resistance property and one-way property for a quantum hash function is also an important task.

Note that the formula for $\ket{\psi(x)}$ gives the classical-quantum function that
transforms a classical input into the quantum state composed of $m$
qudits ($d$-dimensional systems). The same state can be constructed
with appropriate number of 2-dimensional systems (qubits), however
this would imply creating entangled states, which are harder to
create and maintain.

\begin{theorem}
The classical-quantum function $\psi: \mathbb Z_q \to {\cal H}^{d^m}$ given by Eqs. \ref{single-qudit-quantum-hashing}, and \ref{multiqudit-quantum-hashing} is a $\left(\frac{d^m}{q},\varepsilon^m\right)$-resistant quantum hash
function.
\end{theorem}
\begin{proof}
According to
the definition \ref{QHashing_definition}, we need to show two main properties for the function $\psi$:
\begin{enumerate}
    \item $\delta$-one-wayness. 
    The dimension of the input space is $q$, the quantum state space has dimension $d^m$. Thus, $\psi$ is $\delta$-one-way for
    $$\delta=\frac{d^m}{q}.$$
    \item $\varepsilon$-collision-resistance.
    The maximal inner product between unequal quantum hashes is bounded by
    \begin{eqnarray*}
    \fl \max\limits_{x_1\neq x_2}\big|\braket{\psi(x_1)}{\psi(x_2)}\big|=\max\limits_{x_1\neq x_2} \prod_{j=1}^{m}
    \left[\frac{1}{d}\Big|1+e^{i\frac{2\pi s_{j,2}(x_2-x_1)}{q}}+\ldots+e^{i\frac{2\pi s_{j,d}(x_2-x_1)}{q}}\Big|\right]
   \nonumber\\
   = \max\limits_{x\neq 0}\big|\braket{\psi(x)}{\psi(0)}\big|
     \nonumber\\ = \max\limits_{x\neq 0} \prod_{j=1}^{m}\left[\frac{1}{d}\Big|1+e^{i\frac{2\pi s_{j,2}x}{q}}+\ldots+ e^{i\frac{2\pi s_{j,d}x}{q}}\Big|\right]\leq \varepsilon^m,
     \end{eqnarray*}
    if all $S_j=\{s_{j,1}, \ldots, s_{j,d}\}$ are $\varepsilon$-biased sets.
\end{enumerate}
Thus, $\psi$ corresponds to the balanced
$\left(\frac{d^m}{q},\varepsilon^m\right)$-resistant quantum hash
function according to \cite{Ablayev-Ablayev-Vasiliev:2016:Balanced-Quantum-Hashing}.
%, since the probability of correct extraction of $x$ from $\ket{\psi(x)}$ is bounded by $\frac{d^m}{q}$, while the collision probability is bounded by $\varepsilon^m$.
\end{proof}
Note that the proof above also suggests that comparing hashes of two different values $x_1$ and $x_2$ is equivalent to comparing hashes of $x=(x_2-x_1)$ and $0$.

\begin{remark}
Although $\varepsilon$-biased sets give guaranteed collision resistance to our multiqudit quantum hash function, for small sizes of input and output spaces better bounds on collision resistance can be obtained by numeric optimization (see Table \ref{comparison-table} for details). 

At the moment this approach can be used only for relatively small values of $d,m$ and $q$ since it implies an exhaustive search for optimal values of $s_{j,k}$ that give minimum to the following function:
$$\min\limits_{\{s_{j,k}\}}\max\limits_{x\neq 0}\frac{1}{d^{m}} \prod_{j=1}^{m}\Big|1+e^{i\frac{2\pi s_{j,2}x}{q}}+\ldots+ e^{i\frac{2\pi s_{j,d}x}{q}}\Big|.$$

\end{remark}
%Note that the general lower bound suggests that for input state space of size $q$ collision probability bounded by $\varepsilon^m$ implies that a target quantum state space should have size at least $\Omega((\log{q})/\varepsilon^m$). On the other hand, there exist $\varepsilon$-biased sets of size $d=O((\log{q})/\varepsilon)$, and thus our construction gives the output state space of $d^m=O((\log{q})^m/\varepsilon^m)$.

\begin{table}
\caption{\label{comparison-table}The worst-case values of collision probability with parameters from $\varepsilon$-biased sets and from the numeric optimization, $q=256$.}
\begin{tabular}{@{}ccc}
\br
Number of qudits  &
$\varepsilon$-biased sets  &
Exhaustive search\\
\mr
&$d=2$  &\\
\mr
1 &
0,9998 &
0,9998\\

2 &
0,9996 &
0,959 \\
 
3 &
0,9994 &
0,7519 \\

4 &
0,9992 &
0,4378 \\

5 &
0,999 &
0,2031 \\

6  &
0,9988 &
0,0806 \\

7 &
0,9986 &
0,0279 \\
\mr
&$d=3$ &\\
\mr

1 &
0,9681 &
0,9681\\

2 &
0,9372 &
0,5422 \\
 
3 &
0,9073 &
0,1483 \\

4 &
0,8784 &
0,0368 \\

5 &
0,8504 &
0,0063 \\

\mr
&$d=4$&\\
\mr

1 &
0,8329 &
0,8329\\

2 &
0,6937 &
0,2174 \\
 
3 &
0,5778 &
0,0429 \\

4 &
0,4813 &
0,0072 \\

\br
\end{tabular}
\end{table}
\normalsize

\section{Experiment}
\subsection{Single-photon states with an orbital angular momentum}

%The structured light, especially light with an orbital angular momentum (OAM) is a topic of growing interest in the optics community, not only for its inherent properties but also for its possible applications \cite{rev_OAM_1,rev_OAM_2}. For example, the encoding in OAM basis allows us to increase the data capacity of single photon pulses due to the generation high-dimensional states. The achievable data bits per photon increases as $\log_2(N)$, where $N$ represents the number of orthogonal encoding states \cite{143km,Gibson:04,Willner:15,Mirhosseini_2015}. Furthermore, the AOM basis is widely used for space-division multiplexing due to having unique properties of inherent orthogonality and unbounded number of states, in principle. After coaxially propagating in the free space or fiber, beams with different OAM could be efficiently (de)multiplexed using special detection method \cite{multiplex_AOM}. The total system data capacity could be increased by a factor of $N$, where $N$ represents the number of multiplexed OAM beams \cite{terabit_scale_OAM,Pang:18,Terabit_free_space}. In the present paper, we construct a quantum hash as that consists of a sequence of single isolated qudits prepared into the OAM basis. We encoded classical input $x$ into phase of quantum states $\ket{\psi(x)}$. We show that the use of high-dimension states increases the collision resistance of hashing protocols and enhances the robustness against to extraction of information about the classical input $x$ of hashing protocols.
As we have shown above, the construction of the multiqudit quantum hash function (\ref{multiqudit-quantum-hashing}) is adapted for implementation by the sequence of high-dimensional states (qudits) with specific encoding. For that purposes we use single photons with an orbital angular momentum (OAM) generated via spontaneous parametric down-conversion (SPDC). OAM-based encoding is currently widely used for implementing various quantum communication protocols (see, e.g., the reviews \cite{Flamini_2018,willner2021orbital}) and is especially promising for their high-dimensional variants \cite{erhard2018twisted}. 

In the process of SPDC \cite{klyshko1977utilization,hong1986experimental}, when a pump beam propagates through a quadratic nonlinear medium, one of the pump photons spontaneously annihilates and signal and idler photons are simultaneously created. The signal and idler photons have to satisfy the phase matching conditions $\omega_p =\omega_s+\omega_s$ and $\vec{k_p}=\vec{k_s}+\vec{k_i}+\vec{k}_{QPM}$, where $\omega_j$ and $\vec{k_j}$ are the frequency and wave vector, respectively, corresponding to the pump ($j=p$), signal ($j=s$) and idler ($j=i$) photons, and $\vec{k}_{QPM}$ is a so called quasiphasemathing vector ($k_{QPM}=2\pi/\Lambda$) for a crystal poling with period $\Lambda$. If the pump radiation has an orbital angular momentum $l_p$ and all the photons propagate in the same direction (collinear SPDC), then the generated photon pairs also have OAM \cite{Nature_OAM_cons_law, Ibarra-Borja:19} satisfying to the conservation low
\begin{equation}
    l_p = l_s+l_i.
\end{equation}
In the present work, we created single-photon qudit states by projecting the angular momentum of the idler photons onto the mode with $l_i=0$. In this case, since $l_p=l_s$, the spatial structure of the signal photon reproduces that of the pump field. 
%To generate beams with OAM we realized this approach by preparing and measuring OAM carrying beams in the basis of Laguerre-Gaussian $(LG_p^l)$ modes with $p = 0$.
\begin{figure}
    \centering
    \includegraphics[width=9cm]{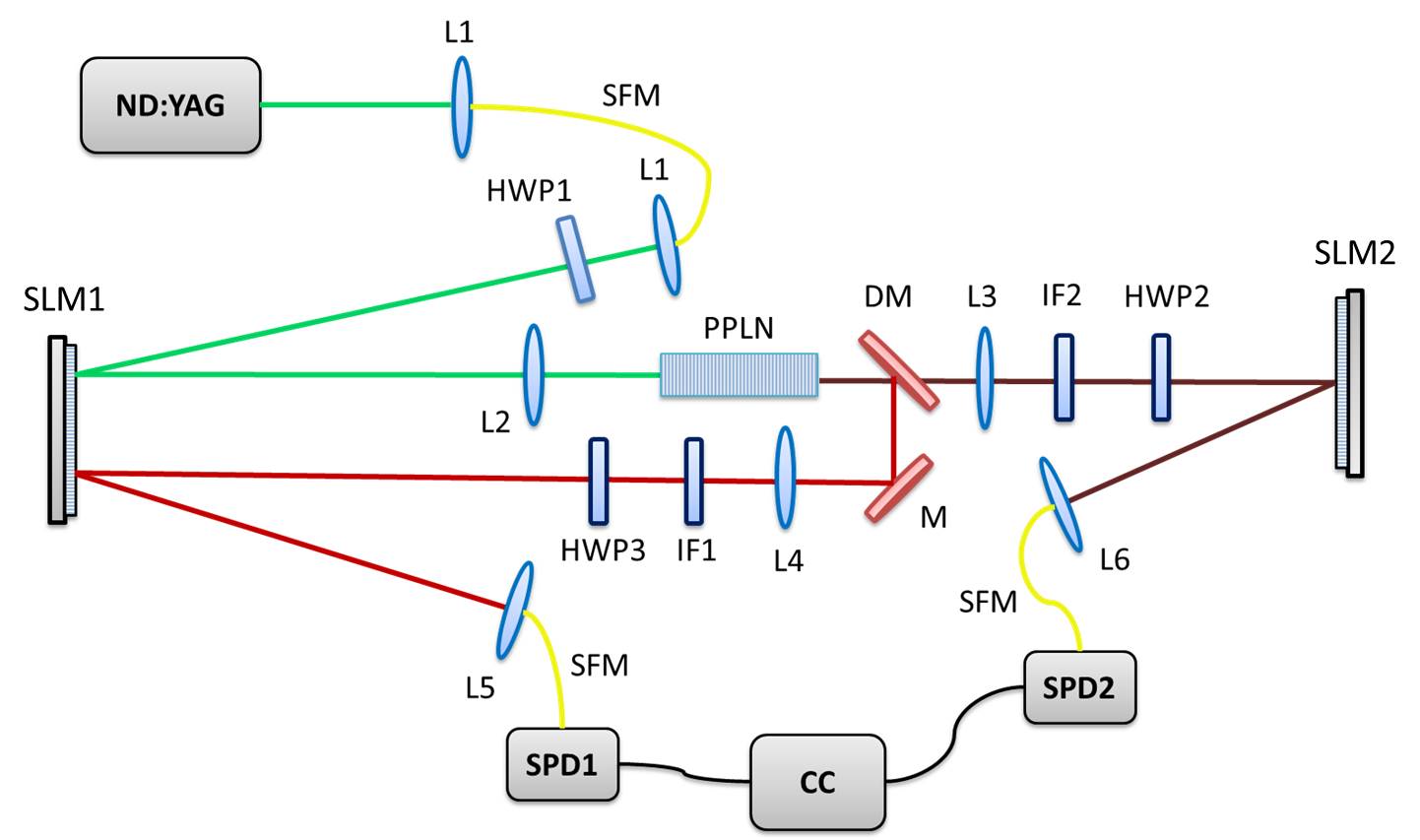}
    \caption{The experimental setup for single-photon qudit generation.}
    \label{fig:experemental setupl}
\end{figure}
To construct the multiqudit version of the quantum hash function we exploit single-photon qudits in a superposition of Laguerre-Gaussian $({\rm{LG}}_p^l)$ modes with the radial index $p = 0$. 
%$d$ modes $\textrm{LG}_0^{\ell_j}$, $j=1, \ldots, d$. 
In particular, we focus on the equally weighted superpositions that have the following form:
\begin{equation}\label{eq:OAM-states)}
\ket{\psi}=\frac{1}{\sqrt{d}}(\ket{\ell_1}+e^{i\varphi_2}\ket{\ell_2}+ \ldots +e^{i\varphi_d}\ket{\ell_d}),
\end{equation}
where $\ket{\ell_n}$ denotes a single-photon state corresponding to $\textrm{LG}_0^{\ell_n}$ mode, $n=1\ldots d$, where $d$ is the dimension of the qudit state space, and $\varphi_n$ is a relative phase.

Our experimental setup is schematically shown in figure 1. We use a 2 cm long type-0 periodically poled LiNbO$_3$:MgO 5$\%$ (PPLN) crystal with a period of 7.50 $\mu$m, which is designed to generate signal photons at wavelength of 810 nm and idler photons at 1550 nm from a pump field at wavelength of 532 nm (CW ND-YAG) under 75\textdegree C operation temperature. The pump field is spatially filtered by a single-mode fiber (SMF) interfaced by two lenses (L1). To prepare required spatial states of the pump field we take advantage of a phase holography technique developed in \cite{Bolduc:13}, where the required modes are obtained after beam reflection from SLM’s screen and picked out the first diffraction order. The first SLM1 (Holoeye PLUTO-2) convert the Gaussian pump beam into an LG mode (part 1 of SLM1) and detect signal photons (part 2 of SLM1) using the phase-flattening technique \cite{Nature_OAM_cons_law}. %The second SLM2 (Holoeye PLUTO-2) detect only LG mode with azimuthal number $\ell = 0$ (???).%
The half-wave plates (HWP$i$, where $i$ = 1, 2, 3) was used to optimize the field polarization with respect to the SLMs. After the preparation of the required LG mode, the pump beam is sent through the PPLN crystal using the biconvex lens L2 with the focal length of 17.5 cm. The signal and idler photons, generated at the wavelength of 810 nm and 1550 nm, respectively, are separated by a dichroic mirror (DM). In the idler arm, the photons are collimated using a lens L3 with the focal length of 150 mm. After that, the idler photons are sent to SLM2, which maintains only $\ell = 0$ mode, and are coupled by the aspheric lens L6 with the focal length of 11 mm into a single mode fiber SMF. As a result, the detected photons have only zero OAM ($\ell_i = 0$). In the signal arm, the photons are collimated using the lens L4 with the focal length of 150 mm and are sent to the part 2 of SLM1. Having been transformed they are coupled by the aspheric lens L5 with the focal length of 11 mm into a SMF. Interference filters IF1 and IF2 are used to select photons at signal and idler wavelengths, respectively. 

The focal length of L2 was chosen to approach the single-Schmidt mode regime of SPDC \cite{Egor_single_Schmidt}. In the single-Schmidt mode regime, the Rayleigh range of the pump beam $z_r = (\pi w^2)/\lambda$ should be equal to half of the crystal length $L$. Following this criteria, we estimated the required pump beam waist to be $w_p = 29\; \mu\textrm{m}$. For the optimal detection of the down-converted modes, the lenses L3 and L4 were chosen so that the Rayleigh range of the signal and idler beams be equal to half of the crystal length $L$, which requires $w_s=36\;\mu\textrm{m}$ and $w_i = 50\;\mu\textrm{m}$. In the experiment we had $w_p = 32\pm 0.7\; \mu\textrm{m}$, $w_s = 40.3\pm 1.2\; \mu\textrm{m}$, and $w_i = 50.1\pm 0.51 \;\mu\textrm{m}$.

The photon detection was carried out using photodetectors SPD1 (SPCM AQR-14, PerkinElmer) operating in the free running mode with an efficiency of 45$\%$, dark count rate of the order of 2 kHz and dead time of 150 ns, and SPD2 (ID Quantique 210) operating in the free running mode with an efficiency of 10\%, dark count rate of the order of 15 kHz, and dead time of 16 $\mu$s. The signals from both detectors are analyzed by the time-to-digital converter (CC, Time Tagger 20). 

To explore quantum hashing main properties we prepared quantum states with the dimension $d = 2$, $3$, and $4$ in the basis of OAM modes $\ell = -3$, $-2$, $\ldots$, $3$. The examples of these states are
\begin{equation}
    \ket{\psi_j^{d=2}(\phi_1)} = \frac{1}{\sqrt{2}}\left(\ket{2} + e^{i\phi_1}\ket{-2}\right),
\end{equation}
\begin{equation}
\ket{\psi_j^{d=3}(\phi_2,\phi_3)} = \frac{1}{\sqrt{3}}\left(\ket{2} + e^{i\phi_2}\ket{-2}+e^{i\phi_3}\ket{0}\right),
\end{equation}
\begin{eqnarray}
 \ket{\psi_j^{d=4}(\phi_4,\phi_5,\phi_5)} = \frac{1}{\sqrt{4}}(\ket{3}+e^{i\phi_4}\ket{-3}+\e^{i\phi_5}\ket{1}  +e^{i\phi_6}\ket{-1}).
\end{eqnarray}

To determine states quality we take advantage of the quantum tomography approach developed in \cite{tomography_1,tomography_2}. As an example, we present here quantum tomographic measurements for the qutrit state $\ket{\psi_j^{d=3}(0,0)}$, which illustrate the accuracy of our experiment. The reconstructed density matrix $\rho_{\ket{\psi_j^{d=3}(0,0)}}^{exp}=\rho_{real}^{exp}+\rho_{imag}^{exp}$ is
\begin{equation}
\fl
    \rho_{real}^{exp}=\left(
    \begin{array}{ccc}
         0.332\pm0.002 & 0.328\pm0.005 & 0.33\pm0.005\\
         0.328\pm0.005 & 0.344\pm0.01 & 0.332\pm0.002\\
         0.33\pm0.005 & 0.332\pm0.002 & 0.344\pm0.009
    \end{array}
    \right),
\end{equation}
\begin{equation}
\fl
    \rho_{imag}^{exp}
    \left(
    \begin{array}{ccc}
        0 & (0.002\pm0.003)i & (0.002\pm0.003)i\\
        (-0.002\mp0.003)i & 0 & -(0.002\pm0.001)i\\
        (-0.002\mp0.003)i & (0.002\pm0.001)i & 0
    \end{array}
    \right).
\end{equation}
The main figure of merit is the fidelity, which is a measure of how close the reconstructed state is to a target state and is given by $F=\left[Tr(\sqrt{ \sqrt{\rho_{target}}\rho^{exp}\sqrt{\rho_{target}} })\right]^2$, where $\rho_{target}$ and $\rho^{exp}$ are the target and reconstructed density matrices, respectively \cite{Fidelety}. We have found that $F_{d=3} = 0.987\pm0.012$. In the same way, we estimated that fidelity for the qubit state $\ket{\psi_j^{d=2}(\frac{2\pi}{3})}$ is $F_{d=2}=0.99\pm0.01$. To determine the purity of these state we calculated the eigenvalues of the density matrices. The largest eigenvalues were $\lambda_{d=2}=0.999$ and $\lambda_{d=3}=0.993$ that correspond to high purity states. As a result, based on the experimental data we can conclude that quantum states of different dimensions are prepared with high purity and high accuracy level. 

\subsection{Implementation of the multiqudit quantum hashing}
In the previous work \cite{Our_qubits_hash}, we have suggested quantum hash functions as a sequence of independent qubits, where classical information was encoded into qubits phase. Now we propose quantum hashing protocol, where the information carriers are high-dimensional states with an orbital angular momentum. In this case, the structure of the quantum hash can be represented by Eqs. \ref{single-qudit-quantum-hashing}, and \ref{multiqudit-quantum-hashing}.
%as
% \begin{equation}
% \ket{\psi_j(x)} = \frac{1}{\sqrt{d}}\left(\ket{\ell_1} + e^{i(2\pi s_{j,2}x/q)}\ket{\ell_2}+\ldots+e^{i(2\pi s_{j,d}x/q)}\ket{\ell_d}\right),
% \end{equation}

Here we propose the implementation of the verification procedure that for a given quantum hash $\ket{\psi(x_1)}$ and a classical value $x_2$ checks whether $x_1=x_2$ or not. The ideal quantum experiment that verifies a multiqubit quantum hash can be set as follows.
\begin{enumerate}
    \item 
We receive a quantum hash of some generally unknown
value $x_1$ as a sequence of $m$ single photons in the overall state $\ket{\psi(x_1)}$:
\begin{equation*}
    \ket{\psi(x_1)}=\ket{\psi_1(x_1)}\otimes\ldots\otimes\ket{\psi_m(x_1)},
\end{equation*}
where the $j$-th qudit is expected to be in the state $\ket{\psi_j(x_1)}$ as described by equation (\ref{single-qudit-quantum-hashing}). 

\item Then we check whether $x_1$ equals to some predefined $x_2$ or not. To do this we perform measurements that project $\ket{\psi_j(x_1)}$ onto $\ket{\psi_j(x_2)}$ and $d-1$ phase orthogonal states
\begin{eqnarray}
\fl
\ket{\psi^{\perp,g}_j(x_2)} = \frac{1}{\sqrt{d}}(\ket{\ell_1} + e^{i\frac{2\pi s_{j,2}x_2}{q}+i\phi_{g,2}}\ket{\ell_2}+\ldots
+e^{i\frac{2\pi s_{j,d}x_2}{q}+i\phi_{g,d})}\ket{\ell_d}),
\end{eqnarray}
where $\big|\braket{\psi_j(x_2)}{\psi^{\perp,g}_j(x_2)}\big|^2=0$, $g = 1, \ldots,d-1$, and $\phi_{g,d}$ are additional phases responsible for orthogonality of the states. For example, if we use the qutrits (3-dimensional states) as information carriers, the set of {$\{\phi_{g,d}\}$} is equal to \{\{$2\pi/3,-2\pi/3$\},\{$-2\pi/3,2\pi/3$\}\}.
\item The projection measurements of these states may be sequential or parallel. In the latter case we have to prepare a complex phase mask on the part 2 of SLM1 which is an appropriate superposition of detection masks for  $\ket{\psi_j(x_2)}$ and $d-1$ states $\ket{\psi^{\perp,g}_j(x_2)}$. The complex mask directs the photons into $d$ detection channels corresponding to the states $\ket{\psi_j(x_2)}, \ket{\psi^{\perp,1}_j(x_2)}, \ldots, \ket{\psi^{\perp,d-1}_j(x_2)}$, respectively. The single-photon detector click in $\ket{\psi_j(x_2)}$ or $\ket{\psi^{\perp,g}_j(x_2)}$ channels corresponds to the outcome $\ket{\psi_j(x_1)}=\ket{\psi_j(x_2)}$ or $\ket{\psi_j(x_1)}\ne \ket{\psi_j(x_2)}$, respectively.

\item If $x_1 = x_2$, the detector of the output $\ket{\psi_j(x_2)}$ would always click, while the other detectors would never click.

\item If $x_1 \ne x_2$, each of the detectors might click, but the probability of erroneous outcome “$x_1 = x_2$” is bounded by the construction of the quantum hash function $\ket{\psi(x)}$.

\item If none of the detectors had clicked, then the qudit is lost, and we either request its resending or tolerate the higher error probability.
 
\item If all of $m$ measurements end up with the outcome "$\ket{\psi_j(x_1)}=\ket{\psi_j(x_2)}$", then the final result of the experiment is considered to be “$x_1 = x_2$”. Otherwise, if at least one qudit leads to $\ket{\psi_j(x_1)}\ne\ket{\psi_j(x_2)}$, then the overall result is also "$x_1 \ne x_2$".

\end{enumerate}

The error probability comes from the fidelity between two different quantum hashes, which is
\begin{equation}\label{eq3}
    \big|\braket{\psi(x_1)}{\psi(x_2)}\big|^2 = \frac{1}{d^{2m}} \prod_{j=1}^{m}\Big|1+e^{i\frac{2\pi s_{j,2}(x_1-x_2)}{q}}+\ldots+ e^{i\frac{2\pi s_{j,d}(x_1-x_2)}{q}}\Big|^2.
\end{equation}
The parameter set $\{s_{j,k}\}$ %is composed of several $\varepsilon$-biased sets $S_j=\{s_{j,1}, \ldots, s_{j,d}\}$ so that $|\braket{\psi(x_1)}{\psi(x_2)}|^2$ is bounded by $\varepsilon^{2m}$. 
is chosen in such a way that the pairs of hashes $\ket{\psi(x_1)}$ and $\ket{\psi(x_2)}$ give the minimal fidelity for $x_1\neq x_2$. 
To measure the collision probability we compare different quantum hashes in the worst-case scenario when quantum hashes for $x_1$ and $x_2$ have the maximum fidelity for a given set $\{s_{j,k}\}$.

The protocol starts with a calibration step on which we adjust the coincidence count rate for the “yes”-answer (“$x_1=x_2$”). We perform about 200 projection measurements of equal states 
%$\big|\braket{\psi(0)}{\psi(0)}\big|^2$ 
and calculate the average coincidence count rate between signal and idler photons. The simultaneous clicks in the idler and signal detectors means that the idler photon has been prepared in the state $\ket{\psi(x_1)}$ and has been successfully projected on the state $\ket{\psi(x_2)}$. 
%We iterate this step about 10 times and pick 
%the average %minimal %!!!
We pick the average 
value as the threshold between “yes” and “no”. In this work, we experimentally evaluate collision probability for multiqudit quantum hash function with $q = 256$ and three groups of parameters: i) $d=2$ and $m=1,\ldots,7$, ii) $d=3$ and $m=1,\ldots,5$, and iii) $d=4$ and $m=1,\ldots,4$, i.e., we perform experiments with various number of qubits, qutrits and ququarts. We encode the classical message $x_1$ into the phase of states considered in Section \ref{quantum-hash-functions}. Since we are considering the worst-case situation, without loss of generality we can focus
%To simplify theoretical and experimental investigations we focused 
on the case when $x_2 = 0$. The value of $x_1$ corresponding to the worst-case scenario was 
%defined from equation (\ref{eq3}) as
calculated from
\begin{equation}
    x_1 = \mathrm{arg}\max\limits_{x\neq0}{\frac{1}{d^{2m}}\prod_{j=1}^{m}\Big|1+e^{i\frac{2\pi s_{j,2}x}{q}}+\ldots+ e^{i\frac{2\pi s_{j,d}x}{q}}\Big|^2}
\end{equation}
for an optimal (quasioptimal) set of parameters $\left\{s_{j,k}\right\}$, which in turn was precomputed as
\begin{equation}
        \left\{s_{j,k}\right\} = \mathrm{arg}\min\limits_{\left\{s_{j,k}\right\}}\max\limits_{x\neq 0}{\frac{1}{d^{2m}}\prod_{j=1}^{m}\Big|1
        +e^{i\frac{2\pi s_{j,2}x}{q}}+\ldots+ e^{i\frac{2\pi s_{j,d}x}{q}}\Big|^2}.
\end{equation}
\begin{figure}
    \centering
    \includegraphics[width=6cm]{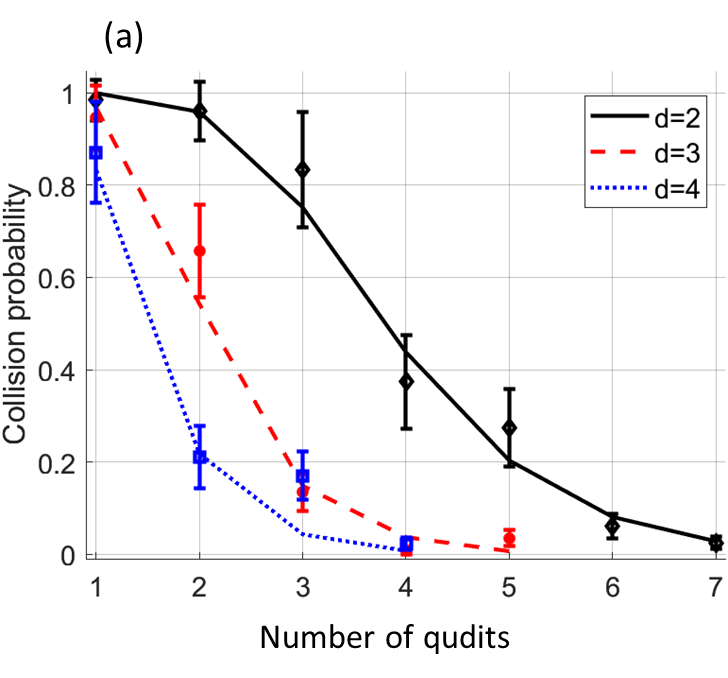} \includegraphics[width=6cm]{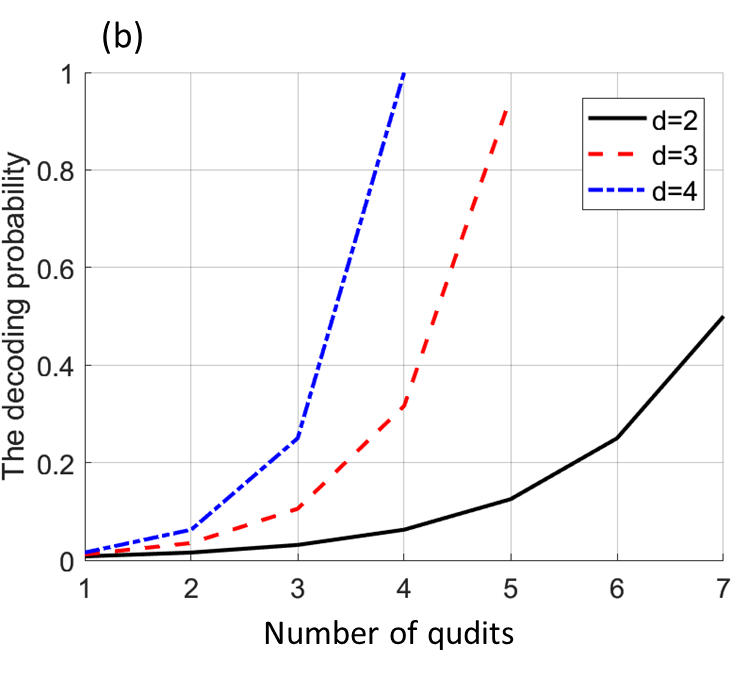}
    \caption{a) Comparing experimental and theoretical collision probabilities for the worst-case scenario for different dimensions of quantum states. b) The theoretical probability of extracting initial message from its quantum hash.}
    \label{fig: Collisions bounds}
\end{figure}

Figure \ref{fig: Collisions bounds} shows the comparison of experimental and theoretical error rates for the worst-case scenario depending on different dimension of quantum states. The experimental results suggest that the proposed technique can be useful even for small sizes of input and output states. Moreover, it can be seen that the number of information carriers decreases with increase of their quantum state dimension for an optimal relation between collision probability and the decoding probability (the probability of extracting the classical input $x$ from the quantum hash). For instance, if we limit the collision probability by 0.25 and the decoding probability by 0.15, the optimal number of qudits to ``compress'' a 8-bit classical information proves to be $m=5$ for $d=2$, $m=3$ for $d=3$, and $m=2$ for $d=4$. Moreover, according to the Holevo theorem \cite{Holevo:1973:bound} no more information can be extracted from a quantum $d$-level system than from a classical $d$-level system, which means that we have a bounded probability (below 1) to extract the information about classical input $x$. 
%For a totally classical system the decoding probability doesn't have the limit, especially if the size of the input space is small.

\section{Conclusion}
In this paper, we have developed a high-dimensional quantum hashing protocol and presented its proof-of-principle implementation using orbital-angular momentum encoding of single photons. Experimental results agree quite well with the
%have confirmed 
theoretical estimations of collision probability dependence on the number of qudits in a quantum hash. 
%Experimental results also validate the technique of high-dimensional states creation and measurement using various dimensions and OAM bases.
An important result is that the number of information  carriers decreases with increase of their quantum state dimension for an optimal ratio between collision probability and the decoding probability. 

Thus, the developed multiqudit quantum hashing approach can speed up 
%both the preparing and checking hash states thereby 
%making them more useful in quantum communications.
quantum communications by reducing the number of particles being transferred while having balanced resistance to inversion and quantum collisions.

%практическое использование в будущемhttps://www.overleaf.com/project/6130d2d08443a048ec3ce88b

\ack
The experimental part of the work was made within the government assignment for FRC Kazan Scientific Center of RAS. The quantum hashing construction and analysis has been supported by the Kazan Federal University Strategic Academic Leadership Program (``PRIORITY-2030'').
The authors express gratitude to \href{https://orcid.org/0000-0003-2943-6337}{Ehsan Shaghaei} for valuable advice and programming support.
%The research is supported by ...%the Russian Science Foundation, project No. 19-19-00656.

\section*{References}

\bibliography{references, references_OAM}

\begin{thebibliography}{10}

\bibitem{Paar-Pelzl:2009:Understanding-Cryptography}
Christof Paar and Jan Pelzl.
\newblock {\em Understanding Cryptography: A Textbook for Students and
  Practitioners}.
\newblock Springer Publishing Company, Incorporated, 1st edition, 2009.

\bibitem{Katz-Lindell:2014:Modern-Cryptography}
Jonathan Katz and Yehuda Lindell.
\newblock {\em Introduction to Modern Cryptography, Second Edition}.
\newblock {CRC} Press, 2014.

\bibitem{476316}
C.H. Bennett, G.~Brassard, C.~Crepeau, and U.M. Maurer.
\newblock Generalized privacy amplification.
\newblock {\em IEEE Transactions on Information Theory}, 41(6):1915--1923,
  1995.

\bibitem{Ablayev-Vasiliev:2014:Crypto-Q-Hashing}
F.M. Ablayev and A.V. Vasiliev.
\newblock Cryptographic quantum hashing.
\newblock {\em Laser Physics Letters}, 11(2):025202, 2014.

\bibitem{Ablayev-et-al:2016:Quantum-Fingerprinting-Hashing}
Farid Ablayev, Marat Ablayev, Alexander Vasiliev, and Mansur Ziatdinov.
\newblock Quantum fingerprinting and quantum hashing. computational and
  cryptographical aspects.
\newblock {\em Baltic Journal of Modern Computing}, 4(4):860--875, 2016.

\bibitem{Ablayev-Ablayev-Vasiliev:2016:Balanced-Quantum-Hashing}
F~Ablayev, M~Ablayev, and A~Vasiliev.
\newblock On the balanced quantum hashing.
\newblock {\em Journal of Physics: Conference Series}, 681(1):012019, 2016.

\bibitem{Our_qubits_hash}
D.~A. Turaykhanov, D.~O. Akat'ev, A.~V. Vasiliev, F.~M. Ablayev, and A.~A.
  Kalachev.
\newblock Quantum hashing via single-photon states with orbital angular
  momentum.
\newblock {\em Phys. Rev. A}, 104:052606, Nov 2021.

\bibitem{Vasiliev:2016:Quantum-Hashing-for-Groups}
Alexander Vasiliev.
\newblock Quantum hashing for finite abelian groups.
\newblock {\em Lobachevskii Journal of Mathematics}, 37(6):751--754, 2016.

\bibitem{Naor-Naor:1990:Small-Bias-Spaces}
Joseph Naor and Moni Naor.
\newblock Small-bias probability spaces: Efficient constructions and
  applications.
\newblock In {\em Proceedings of the Twenty-second Annual ACM Symposium on
  Theory of Computing}, STOC '90, pages 213--223, New York, NY, USA, 1990. ACM.

\bibitem{Alon-Roichman:1994:Random-Cayley-graphs}
Noga Alon and Yuval Roichman.
\newblock Random cayley graphs and expanders.
\newblock {\em Random Structures \& Algorithms}, 5(2):271--284, 1994.

\bibitem{Buhrman:2001:Fingerprinting}
Harry Buhrman, Richard Cleve, John Watrous, and Ronald de~Wolf.
\newblock Quantum fingerprinting.
\newblock {\em Phys. Rev. Lett.}, 87(16):167902, Sep 2001.

\bibitem{Ablayev-Ablayev:2015:Quantum-Hashing-LPL}
Farid Ablayev and Marat Ablayev.
\newblock On the concept of cryptographic quantum hashing.
\newblock {\em Laser Physics Letters}, 12(12):125204, 2015.

\bibitem{Flamini_2018}
Fulvio Flamini, Nicol{\`{o}} Spagnolo, and Fabio Sciarrino.
\newblock Photonic quantum information processing: a review.
\newblock {\em Reports on Progress in Physics}, 82(1):016001, nov 2018.

\bibitem{willner2021orbital}
Alan~E Willner, Kai Pang, Hao Song, Kaiheng Zou, and Huibin Zhou.
\newblock Orbital angular momentum of light for communications.
\newblock {\em Applied Physics Reviews}, 8(4):041312, 2021.

\bibitem{erhard2018twisted}
Manuel Erhard, Robert Fickler, Mario Krenn, and Anton Zeilinger.
\newblock Twisted photons: new quantum perspectives in high dimensions.
\newblock {\em Light: Science \& Applications}, 7(3):17146--17146, 2018.

\bibitem{klyshko1977utilization}
DN~Klyshko.
\newblock Utilization of vacuum fluctuations as an optical brightness standard.
\newblock {\em Soviet Journal of Quantum Electronics}, 7(5):591, 1977.

\bibitem{hong1986experimental}
CK~Hong and Leonard Mandel.
\newblock Experimental realization of a localized one-photon state.
\newblock {\em Physical Review Letters}, 56(1):58, 1986.

\bibitem{Nature_OAM_cons_law}
Alois Mair, Alipasha Vaziri, Gregor Weihs, and Anton Zeilinger.
\newblock Entanglement of the orbital angular momentum states of photons.
\newblock {\em Nature Photonics}, 412:313–316, 2001.

\bibitem{Ibarra-Borja:19}
Zeferino Ibarra-Borja, Carlos Sevilla-Guti\'{e}rrez, Roberto
  Ram\'{i}rez-Alarc\'{o}n, Qiwen Zhan, Hector Cruz-Ram\'{i}rez, and Alfred~B.
  U'Ren.
\newblock Direct observation of oam correlations from spatially entangled
  bi-photon states.
\newblock {\em Opt. Express}, 27(18):25228--25240, Sep 2019.

\bibitem{Bolduc:13}
Eliot Bolduc, Nicolas Bent, Enrico Santamato, Ebrahim Karimi, and Robert~W.
  Boyd.
\newblock Exact solution to simultaneous intensity and phase encryption with a
  single phase-only hologram.
\newblock {\em Opt. Lett.}, 38(18):3546--3549, Sep 2013.

\bibitem{Egor_single_Schmidt}
E.~V. Kovlakov, I.~B. Bobrov, S.~S. Straupe, and S.~P. Kulik.
\newblock Spatial bell-state generation without transverse mode subspace
  postselection.
\newblock {\em Phys. Rev. Lett.}, 118:030503, Jan 2017.

\bibitem{tomography_1}
Megan Agnew, Jonathan Leach, Melanie McLaren, F.~Stef Roux, and Robert~W. Boyd.
\newblock Tomography of the quantum state of photons entangled in high
  dimensions.
\newblock {\em Phys. Rev. A}, 84:062101, Dec 2011.

\bibitem{tomography_2}
B~Jack, J~Leach, H~Ritsch, S~M Barnett, M~J Padgett, and S~Franke-Arnold.
\newblock Precise quantum tomography of photon pairs with entangled orbital
  angular momentum.
\newblock {\em New Journal of Physics}, 11(10):103024, oct 2009.

\bibitem{Fidelety}
S.~P. Walborn, A.~N. de~Oliveira, R.~S. Thebaldi, and C.~H. Monken.
\newblock Entanglement and conservation of orbital angular momentum in
  spontaneous parametric down-conversion.
\newblock {\em Phys. Rev. A}, 69:023811, Feb 2004.

\bibitem{Holevo:1973:bound}
Alexander~S. Holevo.
\newblock Some estimates of the information transmitted by quantum
  communication channel (russian).
\newblock {\em Probl. Pered. Inform. [Probl. Inf. Transm.]}, 9(3):3--11, 1973.

\end{thebibliography}

\end{document}